\documentclass[journal,draftclsnofoot,onecolumn]{IEEEtran}%%Mojtaba
\usepackage{amsmath,amssymb,graphicx,epsfig,cite,fancyhdr,amsthm}
\makeatletter
\newcommand*\dashline{\rotatebox[origin=c]{90}{$\dabar@\dabar@\dabar@\dabar@\dabar@\dabar@\dabar@$}}
\makeatother

\newtheorem{thm}{Theorem}[]
\newtheorem{cor}{Corollary}
\newtheorem{lem}{Lemma}

\theoremstyle{remark}

\theoremstyle{definition}
\newtheorem{defn}{Definition}

\begin{document}

%%%%%%%\title{On the Capacity of the Gaussian Cognitive Z-Interference Channel:\\ Less Noisy and More Capable CGZIC}
\title{Superposition Coding-Based Bounds and Capacity for the Cognitive Z-Interference Channels}

\author{\IEEEauthorblockN{Mojtaba Vaezi and Mai Vu}\\
\IEEEauthorblockA{Department of Electrical and Computer Engineering\\
McGill University\\
Montreal, Quebec H3A 2A7, Canada\\
Email: \{mojtaba.vaezi@mail.mcgill.ca, mai.h.vu@mcgill.ca\}}
%\and
%\IEEEauthorblockN{Mai Vu}
%\IEEEauthorblockA{Department of Electrical and Computer Engineering\\
%McGill University\\
%Montreal, Quebec H3A 2A7, Canada\\
%Email: mai.h.vu@mcgill.ca}
}


\maketitle

%----------------------------------------------------------------------
\begin{abstract}
This paper considers the cognitive interference channel (CIC) with
two transmitters and two receivers, in which the cognitive
transmitter non-causally knows the message and codeword of the
primary transmitter.
We first introduce a discrete memoryless more capable CIC,
which is an extension to the more capable broadcast channel (BC).
Using superposition coding, we propose an inner bound and an outer
 bound on its capacity region. The outer bound is also valid when
 the primary user is under strong interference.
For the Gaussian CIC, this
outer bound applies for $|a| \geq 1 $, where $a$ is the gain of interference
link from secondary user to primary receiver. These capacity inner and
outer bounds are then applied to the Gaussian cognitive Z-interference
channel (GCZIC) where only the primary receiver suffers interference.
Upon showing that jointly Gaussian input maximizes these
bounds for the GCZIC, we evaluate the bounds for this channel.
The new outer bound is strictly tighter than other outer bounds on the capacity
of the GCZIC at strong interference ($a^2 \geq 1 $).
Especially, the outer bound coincides with the inner bound
for $|a| \geq \sqrt{1 + P_1}$ and thus, establishes the capacity of
the GCZIC at this range.
For such an $a$, superposition encoding at the cognitive transmitter
and successive decoding at the primary receiver are capacity-achieving.

\end{abstract}

\IEEEpeerreviewmaketitle

\section{Introduction}\label{sec:intro}
\IEEEPARstart{T}{he} cognitive channel is a special case of an interference channel in
which the second transmitter has complete and non-causal knowledge of
the messages and codewords of the first transmitter. This channel can
be used to model an ideal operating scenario for cognitive radios, a
device that can sense and adapt to the environment intelligently in
coexistence with primary users.
Fundamental limits of such a communication channel are of
interest. Achievable rates of the cognitive channel was first obtained
in \cite{Devroye} by merging
Gel'fand-Pinsker coding \cite{Gel�fand-Pinsker} with the
well-known Han-Kobayashi encoding \cite{Han-Kobayashi} for the
interference channel. At low interference, the capacity region of this
channel in the Gaussian case has recently been established by
\cite{Jovicic-Viswanath} and \cite{Wu-Vishwanath} independently. While
the former considers the Gaussian channel only, the latter studies the
general discrete memoryless channel case, also called the interference
channel with degraded message set (IC-DMS). Cognitive channel capacity
is also known for very strong interference, when both receivers can
decode both messages \cite{Maric1}. At medium interference, the
capacity is still an open problem, with some achievable rate regions
presented in \cite{Jovicic-Viswanath}, \cite{Jiang}, and \cite{Maric2}.

%In \cite{Devroye}, similar to \cite{Han-Kobayashi}, both senders split
%their message into private and public parts, whereas neither
%\cite{Jovicic-Viswanath}  nor \cite{Wu-Vishwanath} does so. On the
%other hand, in \cite{Jiang}, only the cognitive encoder uses rate
%splitting, which allows the primary transmitter to decode a part of
%the interference. The cognitive transmitter also boosts the primary
%user's rate by transmitting the primary user's message and employs
%Gel'fand-Pinsker's binning technique, so that the cognitive receiver
%can decode its message interference-free.

The Z-interference channel (ZIC) is an interference channel in which
only one receiver suffers from interference. Its capacity is also
unknown even for the Gaussian channel, except for some special
cases. From the capacity perspective, it is not important which
transmitter interferes with the other in the ZIC. In a cognitive ZIC,
however, due to asymmetric transmitters, two different ZIC are
conceivable. One is with interference from the cognitive transmitter
to the primary receiver, and the other from the primary transmitter to
the cognitive receiver. While achievable rate regions for the first
one have been studied recently in \cite{Liu-Maric-Z}, \cite{Cao-Z},
there has not been such an investigation for the second one.

In this paper, we study the cognitive channel in general and apply
the results to the Gaussian cognitive ZIC (GCZIC) in
which the cognitive transmitter interferes with the primary
receiver. The contribution can be summarized as
follows.
First, we introduce a new discrete memoryless cognitive
interference channel (DM-CIC) in which the primary receiver is
more capable than the secondary receiver. We term it the more
capable DM-CIC. Then, using superposition coding, we establish
inner and outer bound on its capacity.
We also define a strong interference condition and
show that the proposed outer bound holds under this condition also.
Implicitly,  both inner and outer bounds are also valid
for cognitive Z-interference channel that the interfered
 receiver is more capable than the other receiver.

Second, we show that at strong interference ($a^2 \geq 1$),
where $a$ is the gain of interference
link from secondary user to primary receiver, the outer bound
is applicable to the Gaussian CIC, and thus to the GCZIC.
Then we prove that in Gaussian noise channel, jointly Gaussian distribution
is the optimum distribution for this outer bound; and therefore,
we are able to compute this outer bound for the GCZIC. The outer bound is
proven to be the best outer bound for the GCZIC at strong interference.
%It also becomes tighter as the gain of the interference link increases.

Finally, we derive the Gaussian version of the achievable rate region
and prove that when interference is highly strong, i.e.,
$|a| \geq \sqrt{1 + P_1}$,
the inner and outer bounds coincide.
Thus, we establish the capacity region of the GCZIC at this range, and
show that superposition coding is the capacity achieving scheme.
%Being simplified for substantially strong interference, i.e.,
%$|a| \geq \sqrt{P_1P_2} + \sqrt{1 + P_1 + P_1P_2}$, this region
%can be represented by two inequalities.
For such a large $a$, superposition encoding at the cognitive transmitter
and successive decoding at the primary receiver are capacity-achieving.
%Indeed, we show that for this range of $a$, the channel is
% less noisy GCZIC, that is, the broadcast channel from
% secondary transmitter to primary receiver is less noisy
% than the channel to secondary receiver.
%
%For $1 \leq  a^2 \leq 1 + P_1$, we point out a simple outer bound by
%introducing a less noisy GCZIC with weak interference and using the
%capacity of the GCZIC at this regime.

The rest of paper is organized as follows. In Section \ref{sec:models},
we discuss models for the Gaussian cognitive interference channel
and the GCZIC as well as the existing capacity result for this
channel at $a^2 \leq 1$. We also introduce the more capable DM-CIC
in this section. In Section \ref{sec:DM-CZIC}, we provide new inner
and outer bounds on the capacity
region of the DM-CIC. Then in Section \ref{sec:cap}, we show
that for $a^2 \geq 1$ we can apply the introduced inner and outer bounds
to the GCZIC; and, we compute these bounds for this range. For
$|a|\geq \sqrt{1 + P_1}$, we prove the outer
bound is equal to the proposed achievable region; and thus,
establish the capacity of the GCZIC. Section
\ref{sec:sum} concludes the paper.

%%%%%%%%%%%%%%%%%%%%%%%%%%%%%%%%%%%%%%%%%%%%%%%%%%%%%%%%%%%%%%%%%%%%%%%%%%%%%%%%%%%%%%%%%%%%%%%%%%%%%%%%%%%%%%%%%%%%%%%%%%%%%%%%%%%%%%%%%%%
%%%%%%%%%%%%%%%%%%%%%%%%%%%%%%%%%%%%%%%%%%%%%%%%%%%%%%%%%%%%%%%%%%%%%%%%%%%%%%%%%%%%%%%%%%%%%%%%%%%%%%%%%%%%%%%%%%%%%%%%%%%%%%%%%%%%%%%%%%%

\section{Channel Models, definitions, and existing results}
\label{sec:models}
The classical interference channel (IC) consists of two independent,
non-cooperating pairs of transmitter and receiver, both communicating
over the same channel and interfering each other. A special case of
the IC is the cognitive IC, also called an IC with degraded
message sets (IC-DMS), in which a transmitter, the cognitive one, has
non-causal knowledge of the messages and codewords to be transmitted
by the other transmitter, the primary one. In this section we formally
define this channel and some other derivative of that.

\subsection{Discrete memoryless more capable cognitive interference channel}

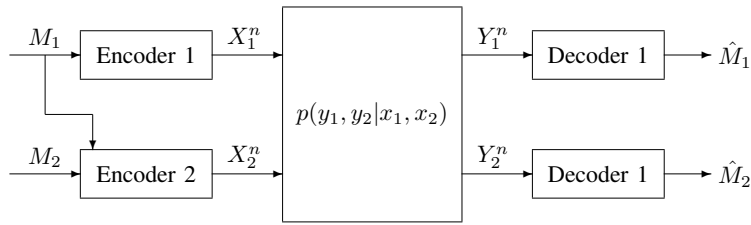
\begin{figure}
\begin{center}
\scalebox {.9}{
\begin{picture}(280,170)

\put(0,50){
\begin{picture}(200,80)
\put(20,20){\framebox(55,20){Encoder 2}}
\put(-10,30){\vector(1,0){30}}
\put(12,37){\makebox(0,0)[r]{$M_{2}$}}
%\put(-55,20){\framebox(45,20){Source 2}}
\put(75,30){\vector(1,0){30}}
\put(96,37){\makebox(0,0)[r]{$X^n_{2}$}}
\put(105,10){\framebox(75,90){$p(y_1,y_2|x_1,x_2)$}}
\put(180,30){\vector(1,0){30}}
\put(200,37){\makebox(0,0)[r]{$Y^n_{2}$}}
\put(210,20){\framebox(55,20){Decoder 1}}
\put(265,30){\vector(1,0){20}}
\put(302,30){\makebox(0,0)[r]{$\hat{M}_{2}$}}
\end{picture}
}

\put(0,100){
\begin{picture}(200,80)
\put(20,20){\framebox(55,20){Encoder 1}}
\put(-10,30){\vector(1,0){30}}
\put(12,37){\makebox(0,0)[r]{$M_{1}$}}
%\put(-55,20){\framebox(45,20){Source 1}}
\put(75,30){\vector(1,0){30}}
\put(96,37){\makebox(0,0)[r]{$X^n_{1}$}}
%\put(105,10){\framebox(75,90){$p(y_1,y_2|x_1,x_2)$}}
\put(180,30){\vector(1,0){30}}
\put(200,37){\makebox(0,0)[r]{$Y^n_{1}$}}
\put(210,20){\framebox(55,20){Decoder 1}}
\put(265,30){\vector(1,0){20}}
\put(302,30){\makebox(0,0)[r]{$\hat{M}_{1}$}}
\put(5,30){\line(0,-1){25}} \put(5,5){\line(1,0){20}}
\put(25,5){\vector(0,-1){15}}
\end{picture}
}
\end{picture}
}
\end{center}

\vspace{-35pt}

\caption{The DM-CIC with two private messages $M_1, M_2$, two inputs
$X_1, X_2$, and two outputs $Y_1, Y_2$. When $p(y_1,y_2|x_1,x_2) =
p(y_2|x_2)p(y_1|x_1,x_2)$, the DM-CIC is converted to the
DM-CZIC.}
  \label{fig:DM-CZIC}
\end{figure}

Consider the discrete memoryless cognitive interference channel
(DM-CIC), also termed the discrete memoryless interference channel
with degraded message sets (IC-DMS), depicted in Figure
\ref{fig:DM-CZIC}, where sender 1 wishes to transmit message $M_1$
to receiver 1 and sender 2 wishes to transmit message $M_2$ to
receiver 2. Message $M_2$ is available only at sender 2, while both
senders know $M_1$. This channel is defined by a tuple $({\cal
X}_1,{\cal X}_2;p(y_1,y_2|x_1,x_2);{\cal Y}_1,{\cal Y}_1)$ where two
inputs ${\cal X}_1,{\cal X}_2$, and two outputs ${\cal Y}_1,{\cal
Y}_1$ are related by a collection of conditional probability density
functions $p(y_1,y_2|x_1,x_2)$.

The discrete memoryless cognitive
Z-interference channel (DM-CZIC) is a DM-CIC in which interference
is one sided. More specifically, we consider the case where the primary
user does not interfere the
secondary one. This only affects the channel transition matrix.
Thus, the DM-CZIC with two private messages $M_1, M_2$, for the two
receivers, two inputs $X_1, X_2$, and two outputs $Y_1, Y_2$ is a
DM-CIC in which
\begin{align}
\label{eq:cond0}
p(y_1,y_2|x_1,x_2) = p(y_2|x_2)p(y_1|x_1,x_2)
\end{align}
 for all $p(x_1,x_2)$.

\begin{defn}
The DM-CIC is said to be more capable if
\begin{align}
\label{eq:defn1}
I(X_1,X_2;Y_1) \geq I(X_1,X_2;Y_2)
\end{align}
 for all $p(x_1,x_2)$.
\end{defn}

\noindent Since the second transmitter can encode and broadcast both messages,
in the absence of the first transmitter this channel reduces to the
well-known more capable DM-BC. In the presence of first sender, this
channel is no longer a BC but an interference channel (IC). However,
due to cognition, the second transmitter has complete and non-causal
knowledge of both messages and codewords; thus, it can act similarly
to the BC's transmitter. This observation motivated us to define  a
condition similar to the one that makes one receiver more capable
than the other one in a DM-BC.

We also define another condition to identify that primary receiver
is in a better situation than secondary receiver in receiving
the signal of cognitive user. We name this strong cognitive inference
condition, as it indicates, roughly speaking, the interference
link from cognitive user to primary reciter is stronger the direct
link of cognitive sender to its corresponding receiver.
\begin{defn}
The DM-CIC is under strong cognitive interference if
\begin{align}
\label{eq:defn2}
I(X_2;Y_1|X_1) &\geq I(X_2;Y_2|X_1)
\end{align}
 for all $p(x_1,x_2)$.
\end{defn}

\noindent Note that in general neither of these two definitions
\eqref{eq:defn1} and \eqref{eq:defn2} implies
the other one.

\subsection{Gaussian cognitive interference channels}

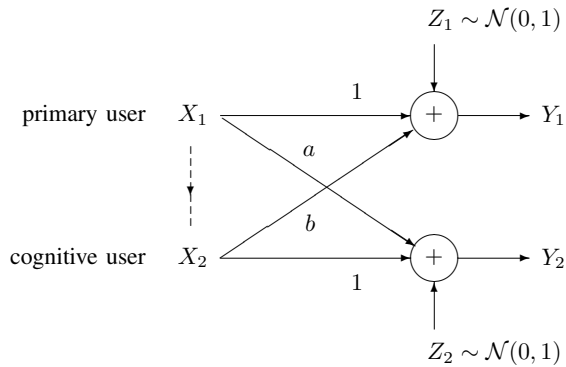
\begin{figure}
\begin{center}
\scalebox{.9}{
\begin{picture}(100,180)
\put(0,50){
\begin{picture}(200,80)
\put(100,10){\circle{20}}
\put(96,8){$+$}
%\put(100,15){\line(0,1){10}}
%\put(95,20){\line(1,0){10}}
\put(10,10){\vector(1,0){80}}
\put(110,10){\vector(1,0){30}}
\put(100,-20){\vector(0,1){20}}
\put(10,10){\vector(3,2){81}}

\thicklines
\put(5,10){\makebox(0,0)[r]{cognitive user $ \quad X_2$}}
\put(70,0){\makebox(0,0)[r]{$1$}}
\put(50,25){\makebox(0,0)[r]{$b$}}
\put(145,10){\makebox(0,0)[l]{$Y_2$}}
\put(97,-30){\makebox(0,0)[l]{$Z_2 \sim {\cal N} (0, 1)$}}
\end{picture}
}

\put(0,100){
\begin{picture}(200,80)
\put(100,20){\circle{20}}
\put(96,18){$+$}
%\put(100,15){\line(0,1){10}}
%\put(95,20){\line(1,0){10}}
\put(10,20){\vector(1,0){80}}
\put(110,20){\vector(1,0){30}}
\put(100,50){\vector(0,-1){20}}
\put(11,19){\vector(3,-2){81}}
\put(-4.5,-12){\dashline}
\put(-2.4,-14){\vector(0,-2){0}}

\put(5,20){\makebox(0,0)[r]{primary user $\quad X_1$}}
\put(70,30){\makebox(0,0)[r]{$1$}}
\put(50,5){\makebox(0,0)[r]{$a$}}
\put(145,20){\makebox(0,0)[l]{$Y_1$}}
\put(97,60){\makebox(0,0)[l]{$Z_1 \sim {\cal N} (0, 1)$}}
\end{picture}
}
\end{picture}
}
\end{center}
\vspace{-15pt}
\caption{ The Standard Gaussian cognitive channel with inputs $ X_{1},
  X_{2},$ outputs $ Y_{1},  Y_{2},$ and additive noises $ Z_{1},
  Z_{2}.$ The dashed line shows the flow of non-causal information
  from the primary user to the cognitive user.}
  \label{fig:StandardIC}
\end{figure}

Without loss of generality, we use the standard form of the Gaussian
interference channel \cite{Carleial}, \cite{Kramer}, in which the
gains of both direct links are 1 and both noises are independent with
unit variance. The standard Gaussian cognitive interference channel
is shown in Figure \ref{fig:StandardIC} and is expressed as
\begin{align*}
  Y_1 &= X_1 + bX_2 + Z_1 \\
  Y_2 &= aX_1 + X_2 + Z_2
\end{align*}
Here the interference links are arbitrary constants $a$ and $b$ known
at all the transmitters and receivers; $ X_{1},\; X_{2}$  represent
the primary and secondary users' transmit signals, and $Y_{1},\;
Y_{2}$ their received signals; $Z_1, Z_2$ are independent additive
noises $Z_{i} \sim {\cal N}(0,1)$ ($i=1,2$). We also assume that
transmitted signals are subject to average power constraint as $
E[X_{1}^2] \leq P_{1} $ and $E[X_{2}^2] \leq P_{2} $.

Depending on the values of the interference links $a$ and $b$,
different classes of IC emerge. A special class is the Z-interference
channel (ZIC) when either $a=0$ or $b=0$. For a non-cognitive system,
there is no difference in the capacity analysis of these two ZICs. In
a cognitive system, however, due to asymmetric knowledge at the
transmitters, two different cognitive ZICs are conceivable. One is
when the primary receiver has no interference ($a=0$), and the other
is when the secondary receiver has no interference ($b=0$). These two
GCZIC channels have completely different capacity regions.

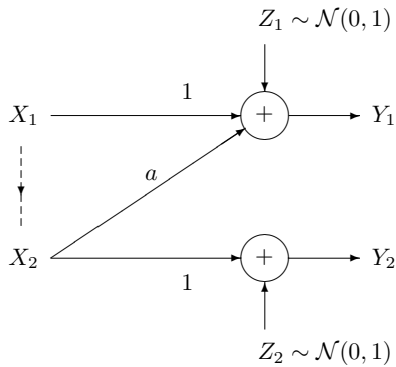
\begin{figure}
\begin{center}
\scalebox{.9}{
\begin{picture}(150,180)

\put(0,50){
\begin{picture}(200,80)
\put(100,10){\circle{20}}
\put(96,8){$+$}
\put(10,10){\vector(1,0){80}}
\put(110,10){\vector(1,0){30}}
\put(100,-20){\vector(0,1){20}}
\put(10,10){\vector(3,2){82}}

\thicklines
\put(5,10){\makebox(0,0)[r]{$X_2$}}
\put(70,0){\makebox(0,0)[r]{$1$}}
\put(55,45){\makebox(0,0)[r]{$a$}}
\put(145,10){\makebox(0,0)[l]{$Y_2$}}
\put(97,-30){\makebox(0,0)[l]{$Z_2 \sim {\cal N} (0, 1)$}}
\end{picture}
}

\put(0,100){
\begin{picture}(200,80)
\put(100,20){\circle{20}}
\put(96,18){$+$}
\put(10,20){\vector(1,0){80}}
\put(110,20){\vector(1,0){30}}
\put(100,50){\vector(0,-1){20}}
%\put(10,20){\vector(3,-2){82}}
\put(-4.5,-12){\dashline}
\put(-2.4,-14){\vector(0,-2){0}}

\put(5,20){\makebox(0,0)[r]{$X_1$}}
\put(70,30){\makebox(0,0)[r]{$1$}}
%\put(50,5){\makebox(0,0)[r]{$a$}}
\put(145,20){\makebox(0,0)[l]{$Y_1$}}
\put(97,60){\makebox(0,0)[l]{$Z_1 \sim {\cal N} (0, 1)$}}
\end{picture}
}
\end{picture}
}
\end{center}
\vspace {-15pt}
\caption{ The Gaussian cognitive Z-interference channel (GZIC)}
 \label{fig:ZIC}
\end{figure}

The capacity of the GCZIC with $a = 0$ can be simply obtained from the
well-known result of dirty paper coding by Costa
\cite{Costa}. Achievable rate and capacity regions of this cognitive
ZIC for the discrete memoryless case can also be found in
\cite{Liu-Maric-Z}. On the other hand, to the best of our knowledge,
not much work has been done on the second GCZIC (with $b=0$). In this
paper, we investigate the capacity region for this GCZIC with
$b=0$. In the rest of this paper, GCZIC refers to this channel.

In the following sections, we establish an inner bound and an outer
 bound for the DM-CIC satisfying either \eqref{eq:defn1} or \eqref{eq:defn2}.
These bounds are valid for the DM-CZIC as well.
Later, we will use these bounds to prove the capacity of the
GCZIC with very strong interference.

\section{Inner and Outer bounds on the capacity of the more capable DM-CIC}
\label{sec:DM-CZIC}

In the first part of this section, we derive an achievable rate region
for the DM-CIC. In the second part, we introduce a new
outer bound on the capacity of the more capable DM-CIC which is valid for
 DM-CIC with strong interference also.
Since the more capable DM-CIC is an extension of the more capable DM-BC,
the achievable region also is an extension of its component's capacity region. The
technique used for achievability is the same as capacity achieving
technique for the conventional
more capable DM-BC. In addition, the outer bound also resembles that of the
more capable DM-BC. Similarly, achievability, error analysis, and
the proof of converse (outer bound) follow those of the DM-BC.
Nevertheless, in general the inner bound and the outer bound are not equal
in the more capable DM-CIC while they are proven to be the same for the
more capable DM-BC. Indeed, this difference, which will be addressed
later in this section, prevents establishing capacity region for the more capable
DM-CIC.

\subsection{A new achievable rate region }
\label{inner}
Theorem 1 provides an achievable region for the DM-CIC.
The achievable technique uses superposition encoding at the
cognitive transmitter. The decoding is based on the joint
typicality.

\begin{thm}
An achievable rate region for the DM-CIC consist of all
rate pairs $(R_{1},R_{2})$ that satisfy
\begin{align}
  R_1 &\leq I(X_1;Y_1|U)   \nonumber \\
  R_2 &\leq I(U;Y_2)  \nonumber \\
  R_1 + R_2 &\leq I(X_1,X_2;Y_1)
  \label{eq:inner}
\end{align}
for some joint distributions that factors as
$p(u)p(x_1)x_2(x_1, u)p(y_2|x_2)p(y_1|x_1,x_2)$, where
$x_2(x_1, u)$ is a function which can be random or deterministic.
\label{thm1}
\end{thm}

\noindent The proof uses the superposition coding idea in which
$Y_2$ can only decode $M_2$ (the cloud center)
while $Y_1$ is intended to decode the satellite codeword.
For completeness we provide the proof in the appendix A.
%considering the space of all codewords, one can
%view the $U$ as cloud centers, and the $X_1$ as satellites.

\subsection{More capable BC capacity inspired outer bound}
Inspired by capacity of more capable BC \cite{ElGamalMoreCapable}, \cite{ElGamal},
instead of proving the outer bound for region \eqref{eq:inner},
we prove it for the slightly altered rate region below.
The following outer bound on the capacity holds both for the more capable DM-CIC and
DM-CIC with strong interference.
\begin{thm}
The union of all rate pairs $(R_{1},R_{2})$ such that
\begin{align}
  R_2 &\leq I(U;Y_2)  \nonumber \\
  R_1 + R_2 &\leq I(X_1;Y_1|U) + I(U;Y_2)   \nonumber \\
  R_1 + R_2 &\leq I(X_1,X_2;Y_1)
 \label{eq:Outer}
\end{align}
for some joint distributions $p(u,x_1,x_2)p(y_1,y_2|x_1,x_2)$
constitutes an outer bound on the capacity region of a DM-CIC
satisfying either the more capable condition in \eqref{eq:defn1} or
strong interference condition in \eqref{eq:defn2}.
\label{thm2}
\end{thm}

\noindent The proof is based on the proof of converse for the more
capable BC in \cite{ElGamalMoreCapable} but adapted for the DM-CIC.
For completeness, we provide the proof in the appendix B.
In the BC, this new form is shown to be an alternative representation
of the rate region in Theorem \ref{thm1}  \cite{ElGamalMoreCapable};
thus, proving the converse for this equivalent region
establishes the capacity of the more capable BC.
However, these two regions are not equivalent for DM-CIC
because of different input distributions.
Therefore, Theorem \ref{thm2} provides only
an outer bound for the capacity of the more capable DM-CIC and DM-CZIC. Nevertheless,
later in this paper we show that this outer bound is tight for the GCZIC
at very strong interference.

\section{Bounds and capacity of the Gaussian cognitive Z channel at strong interference}
\label{sec:cap}

%\subsection{Existing capacity results for the GCZIC}
The GCZIC at weak interference ($a^2\leq 1$)
is a special case of the
Gaussian cognitive interference channel (when $b=0$), for which the
capacity region is known for $a\leq 1$ and any real $b$
\cite{Jovicic-Viswanath}. The cognitive user partially devotes its
power to help send the codeword of the
primary user. It dirty paper encodes its own codeword
against the codeword of the primary user. The
cognitive receiver performs dirty paper decoding to extract
its message free of interference \cite{Jovicic-Viswanath} and
\cite{Wu-Vishwanath}. At strong interference regime ($a^2\geq 1$) however,
the capacity of the GCZIC is not known in general. An outer bound on the
 capacity of the Gaussian cognitive IC was established by
Maric et al. in \cite{Maric1}, Corollary 1.

In this section, we first find the condition in which the GCZIC
is a more capable or under strong interference.
We show that for the Gaussian CIC, strong interference conditions
is equivalent to $a^2\geq 1$. Thus for
$a^2\geq 1$, we provide new inner and outer bounds by evaluating the inner
and outer bounds in Section \ref{sec:DM-CZIC}. Finally,
we prove that these inner and outer bounds coincide when the interference
is very strong ($|a| \geq \sqrt{1 + P_1 }$),
thus establish the capacity of the GCZIC for
this range of interference.

\subsection{More capable and strong interference conditions for the GCZIC }
\label{strongIC}

In this section we explore the conditions for which Theorem~\ref{thm2}
holds for the GCZIC; i.e, we find the condition that the GCZIC is
either more capable \eqref{eq:defn1} or under strong interference
\eqref{eq:defn2}.

Intuitively, the GCZIC is more capable when interference
is very strong. considers the equivalent channel in
Figure \ref{fig:ZIC-Eq} which is achieved by
manipulating Figure \ref{fig:ZIC}.
 Since both figures have the same $Y_2$, and $Y_1$ is
a scaled transformation of $\tilde{Y}_1$ $(Y_1 = a \tilde{Y}_1)$,
the channels depicted in these figures are equivalent from capacity
point of view. The equivalent channel in Figure \ref{fig:ZIC-Eq}
looks like a broadcast channel if we consider $X_1/a$ as
interference. Without $X_1/a$ this channel is a degraded BC and its
capacity is known. Now, considering the interference $X_1/a$
as noise, and assuming that $a$ is large enough that the
power associated with noise plus interference ($Z_1/a + X_1/a$) is
less than noise power at $Y_2$, then $\tilde{Y}_1$ can be more
capable than $Y_2$.

\begin{figure}[!tb]
\begin{center}
\scalebox{.9}{
\begin{picture}(200,180)

\put(0,50){
\begin{picture}(200,80)
\put(100,10){\circle{20}} \put(96,8){$+$}
\put(10,10){\vector(1,0){80}} \put(110,10){\vector(1,0){30}}
\put(100,-20){\vector(0,1){20}} \put(10,10){\vector(3,2){81}}

\thicklines \put(5,10){\makebox(0,0)[r]{$X_2$}}
\put(70,0){\makebox(0,0)[r]{$1$}} \put(55,45){\makebox(0,0)[r]{$1$}}
\put(145,10){\makebox(0,0)[l]{$Y_2$}}
\put(97,-30){\makebox(0,0)[l]{$Z_2$}}
\end{picture}
}

\put(0,100){
\begin{picture}(200,80)
\put(100,20){\circle{20}} \put(96,18){$+$}
\put(110,20){\vector(1,0){30}} \put(160,20){\vector(1,0){30}}
\put(150,20){\circle{20}} \put(146,18){$+$}
\put(100,50){\vector(0,-1){20}}\put(150,50){\vector(0,-1){20}}

%\put(70,30){\makebox(0,0)[r]{$1$}}
%\put(50,5){\makebox(0,0)[r]{$a$}}
\put(195,20){\makebox(0,0)[l]{$\tilde{Y}_1$}}
\put(97,60){\makebox(0,0)[l]{$Z_1/a$}}
\put(147,60){\makebox(0,0)[l]{$X_1/a$}}
\end{picture}
}
\end{picture}
}
\end{center}

\vspace{-8 pt}
 \caption{ The equivalent channel to the GCZIC in Figure \ref{fig:ZIC}
, with inputs $X_1, X_2,$ outputs $ \tilde{Y}_{1}= Y_1/a,  Y_{2},$
and effective additive Gaussian noises
 $ Z_{1}/a $ and $ Z_{2}.$}
  \label{fig:ZIC-Eq}
\end{figure}
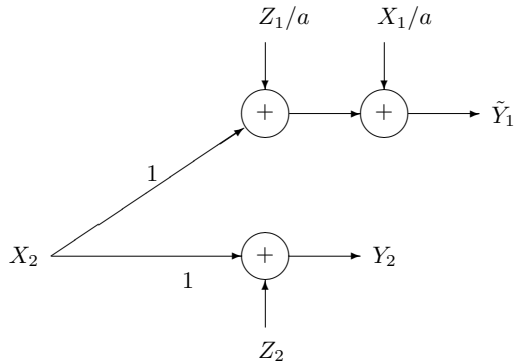
We need to find the range of $a$ for which $\tilde{Y}_{1}$ in Figure~\ref{fig:ZIC-Eq} (or
equivalently $Y_1$ in Figure~\ref{fig:ZIC}) is more capable than $Y_2$ in decoding $X_2$.
The condition (\ref{eq:defn1}) is equivalent to $ I(X_2, X_1;Y_1) \geq I(X_2;Y_2)$
because of channel transition matrix at
(\ref{eq:cond0}). For the GCZIC, this is equivalent to
\begin{align}
 h(X_1+aX_2+Z_1)&\geq h(X_2+Z_2) \label{eq:cond2}
\end{align}
for all $p(x_1,x_2) $.
With jointly Gaussian $X_1,X_2$ with correlation factor $\rho_{12}$
then \eqref{eq:cond2} becomes
\begin{align}
P_1 +a^2P_2 + 2a\rho_{12}\sqrt{P_1 P_2} +1&\geq P_2 +1.  \nonumber
\end{align}
Choosing $\rho_{12}=-1$, which is the worst case, the GCZIC
is more capable if $|a|\geq 1 + \sqrt{P_1 /P_2}.$ It is not clear, however, if
this condition implies more capability.

Let's now find the range of $a$ for which the strong interference
condition \eqref{eq:defn2} holds for the GCZIC. The proof follows
directly from the strong interference condition for the Gaussian
cognitive IC \cite{ElGamal}, which shows that condition \eqref{eq:defn2}
is equivalent to
\begin{align}
%I(X_2;Y_1|X_1) &\geq I(X_2;Y_2|X_1) \nonumber \\
%\Longleftrightarrow h(Y_1|X_1) - h(Y_1|X_1,X_2) &\geq h(Y_1|X_2) - h(Y_1|X_1,X_2)\nonumber \\
%\Longleftrightarrow h(Y_1|X_1) &\geq h(Y_1|X_2)\nonumber \\
%\Longleftrightarrow
h(aX_2+Z_1|X_1)&\geq h(X_2+Z_2|X_1)  \nonumber \\
%\Longleftrightarrow h(aX_2+Z_1,X_1)&\geq h(X_2+Z_2,X_1)  \nonumber \\
%\Longleftrightarrow a^2P_2+1&\geq P_2+1  \nonumber \\
\Longleftrightarrow a^2 &\geq 1.
\end{align}
We can draw the conclusion that the strong interference condition
\eqref{eq:defn2} implies the more capability condition \eqref{eq:defn1}
for any Gaussian cognitive IC. This completes the proof that
Theorem~\ref{thm2} is applicable for the Gaussian CIC, and thus for the GCZIC, if $|a| \geq 1$.

\subsection{A new outer bound on the capacity of the GCZIC at strong interference}
\label{outer}
 The capacity of the GCZIC is partially unknown
for strong interference ($a^2 \geq 1$). At this regime, the best outer bound on
the capacity the GCZIC was established in \cite{Maric3}, Corollary 1.
In this section, we provide a new outer
bound for the capacity of the GCZIC at strong interference.
This outer bound is the Gaussian version of the outer bound in
Theorem \ref{thm2}, with the extra inequality
$R_2 \leq I(X_2;Y_2|X_1)$ \cite {Wu-Vishwanath},
\cite {Maric1} to that.

\begin{lem}
Any achievable rate pair $(R_1, R_{2})$ of the GCZIC with $a^2 \geq 1$, is
upper bounded by the following constraints
 \begin{align}
  R_2 &\leq  \frac{1}{2}\log{\left( \frac{1+ P_2}{ 1 +P_2(1- \rho_{2}^2)} \right)}
  \label{eq:O01} \\
 % R_1 + R_2 &\leq  \frac{1}{2}\log{\left( 1+(1- \rho_{1}^2) P_1 + a^2 (1-\rho_{2}^2 ) P_2 +
%2|a| \sqrt {(1-\rho_{1}^2 )(1-\rho_{2}^2 )P_1P_2} \right)} \nonumber \\
%   &+
%  \frac{1}{2}\log{\left( \frac{1+ P_2}{ 1 +P_2(1- \rho_{2}^2)} \right)}
%  \label{eq:O02} \\
   R_1 + R_2 &\leq  \frac{1}{2}\log{\left( 1+ \big(\sqrt {(1- \rho_{1}^2) P_1} +|a| \sqrt {(1-\rho_{2}^2 )P_2} \big)^2\right)} + \frac{1}{2}\log{\left( \frac{1+P_2}{ 1 +P_2(1- \rho_{2}^2)} \right)}  \label{eq:O02}\\
    R_1 +R_{2} &\leq \frac{1}{2}\log{ \left( 1+P_1 + a^2 P_2 + 2|a|\sqrt{\rho_{12}^2P_1P_2} \right)}
     \label{eq:O03}\\
     R_2 &\leq  \frac{1}{2}\log{\left(  1 +(1-\rho_{12}^2)P_2 \right)}
       \label{eq:O04}
 \end{align}
  where $|\rho_{12} - \rho_{1}\rho_{2}|\leq \sqrt {(1-\rho_{1}^2 )(1-\rho_{2}^2 )}$, and $|\rho_{i}| \leq 1, i=1,2$.
\label{lem1}
\end{lem}

\noindent The proof of this lemma involves  showing that the jointly Gaussian distribution
is the optimum distribution and evaluating the outer bound in Theorem \ref{thm2}
for the GCZIC, then finding the covariance matrix of jointly Gaussian $X_1, X_2, U$
to maximize the RHS of all inequalities in Theorem \ref{thm1}. Details of
evaluation and maximization can be found in Appendix C.

\begin{figure}
  \centering
  \includegraphics [scale=0.7] {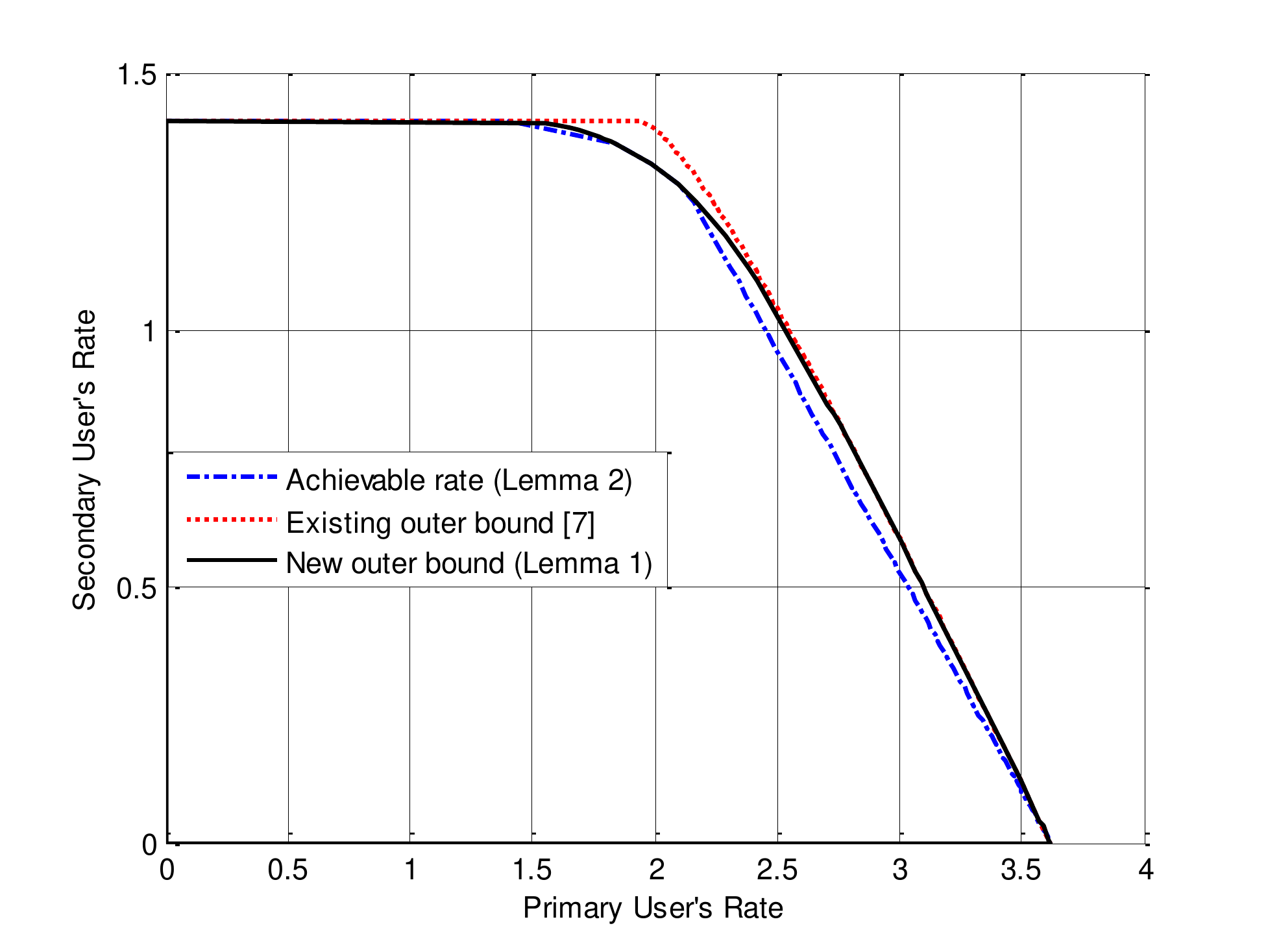}
  %\centering{\epsfig{figure=IC.eps, width=8cm}}
  \caption{ Comparison of new outer bound in Lemma~\ref{lem1}
  and the best existing outer bound in \cite{Maric3}, for the GCZIC
  with $P_1=P_2=6, |a| = 4$.}
  \label{fig:Fig1}
\end{figure}

From the proof, it can be seen that without \eqref{eq:O04}, optimality condition is
$|\rho_{12} - \rho_{1}\rho_{2}|= \sqrt {(1-\rho_{1}^2 )(1-\rho_{2}^2 )}$ which is achieved
when $h(X_2|U,X_1)=0$, and implies that  $X_2$ is a function of $(U,X_1)$. This condition
is the same as that in the inner bound. However, the inner bound applies to independent
$X_1$ and $U$ whereas the outer bound is for general $X_1$ and $U$.
In Section~\ref{capacity}, we show that for a certain range of $a$ ($a^2 \geq 1+P_1$) the optimal
$X_1$ and $U$ for the outer bound are also independent, thus establish the capacity 
at that range.

Figure~\ref{fig:Fig1} numerically compares the outer bound proposed
 in Lemma~\ref{lem1} with the best existing outer bound in \cite{Maric3}.
It shows that new outer bound is strictly better than the existing one.
Moreover, as $|a|$ becomes larger, the gap between these two bounds increases,
and the outer bound in Lemma~\ref{lem1} gets much closer to the
achievable region.
This lemma establishes the best outer bound on the capacity region of the
GCZIC at strong inference ($a^2 \geq 1$). We prove this claim by
simplifying the outer bound in Lemma~\ref{lem1} into three simpler outer bounds
in the following Corollaries. To do so, we first introduce
 $\alpha = 1- \rho_{1}^2, \beta= 1- \rho_{2}^2, \gamma= \rho_{12}^2$,
 $\bar{x}=1-x$, and $x\in [0, 1]$ to make the bounds easier to read.
These Corollaries give the capacity of
the desired channel for two disjoint ranges of $a$.

\begin{cor}
Any achievable rate pair $(R_1, R_{2})$ of the GCZIC, is
upper bounded by the convex hull of the following region
 \begin{align}
   R_1 +R_{2} &\leq \frac{1}{2}\log{ \Big(1+P_1 + a^2 P_2 + 2 |a|\sqrt{\gamma P_1P_2}
   \Big) }
     \nonumber \\
     R_2 &\leq \frac{1}{2}\log{\left(1 + \bar{\gamma}P_2 \right)}
  \end{align}
   where $\gamma \in [0, 1]$.
\label{cor1}
\end{cor}

\noindent This follows immediately by removing \eqref{eq:O01}, \eqref{eq:O02} from the Lemma~\ref{lem1}.
This is the same as the best existing outer bound in \cite{Maric3}, and our
claim that Lemma~\ref{lem1} provides the best outer bound follows readily. It should
be highlighted that, this region has been recently proven in \cite{Rini} and \cite{JiangZ}
to be the capacity of the GCZIC for $1 \leq a \leq  \sqrt {1 + \frac{P1}{1+P2}}$.
%A second special case of Lemma~\ref{lem1} is presented in the following corollary.
\begin{cor}
Any achievable rate pair $(R_1, R_{2})$ of the GCZIC, is
upper bounded by the convex hull of the following region
 \begin{align}
  R_2 &\leq \frac{1}{2}\log{\Big(1 +
    \frac{\bar{\alpha}P_2}{1+\alpha P_2}\Big)}
  \nonumber\\
  R_1 + R_2 &\leq  \frac{1}{2}\log{\Big(1+(\sqrt{P_1} + |a| \sqrt{\alpha P_2}
    )^2\Big)} + \frac{1}{2}\log{\Big(1 +
    \frac{\bar{\alpha}P_2}{1+\alpha P_2}\Big)}
   \end{align}
   where $\alpha \in [0, 1]$.
\label{cor2}
\end{cor}

\noindent Similar to the Corollary~\ref{cor1}, we first eliminate \eqref{eq:O03}, \eqref{eq:O04}
in the Lemma~\ref{lem1}. Then, the proof of is immediate by looking at
\eqref{eq:O02}, and is achieved for $\rho_1=0 \, (\beta = 1)$.
As we will show later in Section~\ref{capacity}, Corollary~\ref{cor2} is also the capacity region for the GCZIC when
interference gain is considerably large ($|a|\geq \sqrt{ P_1P_2} + \sqrt{1 + P_1 + P_1P_2}$).

%\noindent Lastly, a third special case of Lemma~\ref{lem1} is obtainable as shown in Corollary~\ref{cor3} below.
\begin{cor}
Any achievable rate pair $(R_1, R_{2})$ of the GCZIC, is
upper bounded by the convex hull of the following region
 \begin{align}
  R_2 &\leq \frac{1}{2}\log{\Big(1 +
    \frac{\bar{\alpha}P_2}{1+\alpha P_2}\Big)}
 \nonumber \\
  R_1 + R_2 &\leq  \frac{1}{2}\log{\Big(1+(\sqrt{\beta P_1} + |a| \sqrt{\alpha P_2}
    )^2\Big)} + \frac{1}{2}\log{\Big(1 +
    \frac{\bar{\alpha}P_2}{1+\alpha P_2}\Big)}
  \nonumber \\
    R_1 +R_{2} &\leq \frac{1}{2}\log{ \Big(1+P_1 + a^2 P_2 + 2|a|  (\sqrt{\alpha \beta} + \sqrt{\bar{\alpha} \bar{\beta}}) \sqrt{ P_1P_2}
   \Big) }
     \label{eq:O05}
  \end{align}
   where $\alpha, \beta \in [0, 1]$.
\label{cor3}
\end{cor}

\noindent To prove this, we first remove \eqref{eq:O04} in the Lemma~\ref{lem1}.
Then, in \eqref{eq:O03} we use $a \rho_{12} \leq |a||\rho_{12}| \leq |a|(|\rho_{1}\rho_{2}| + \sqrt {(1-\rho_{1}^2 )(1-\rho_{2}^2 )})$
where the last inequality follows applying triangle inequality $|x-y| \geq |x| - |y|$ with
  $x= \rho_{12}$ and $y= \rho_{1}\rho_{2}$ to the LHS of $|\rho_{12} - \rho_{1}\rho_{2}|\leq \sqrt {(1-\rho_{1}^2 )(1-\rho_{2}^2 )}$. The maximum is attained when $\rho_{1}\rho_{2}$ has the same sign with $a$.
Finally, the outer bound in Corollary~\ref{cor3} is obtained considering that $\alpha \triangleq 1- \rho_{2}^2$, $\beta \triangleq 1- \rho_{1}^2$.
Similar to Corollary~\ref{cor2}, in Section~\ref{capacity}, we will show that for $\beta = 1$, Corollary~\ref{cor3} results in the capacity region of the GCZIC when $|a|\geq \sqrt{1 + P_1}$.

\subsection{Superposition coding-based inner bound for the GCZIC}
\label{inner_G}
In this section, we compute the Gaussian version of the achievable region
introduced in Section \ref{inner} for the GCZIC.
Following lemma, which extends Theorem \ref{thm1} to the GCZIC, provides
the achievable region by superposition coding.
\begin{lem}
Any rate pair $(R_1, R_{2})$ satisfying
 \begin{align}
   R_1 &\leq \frac{1}{2}\log{ \left (1 + (\sqrt{P_1} + a \sqrt{\alpha P_2 })^2 \right)}
   \nonumber \\ %\label{eq:PrimaryCorollary1} \\
   R_2 &\leq \frac{1}{2}\log{\left (1+ \frac{\bar{\alpha} P_2}{1 + \alpha P_2} \right)}
   \nonumber \\ %\label{eq:PrimaryCorollary1}
   R_1 + R_2 &\leq \frac{1}{2}\log{\Big(1+ P_1 + a^2 P_2 + 2 a \sqrt{\alpha P_1P
      _2}\Big)}
 \end{align}
with $\alpha \in [0, 1]$ is achievable for the GCZIC.
\label{lem2}
\end{lem}

\begin{proof}
The achievability of this region is straightforward by Theorem \ref{thm1}.
In the proof of Lemma~\ref{lem1}, we have shown that jointly
Gaussian input is the optimum input to maximize
%each of the terms that appear in the RHS of the Theorem \ref{thm1}, i.e.,
$I(X_1;Y_1|U), I(U;Y_2), $ and $I(X_1,X_2;Y_1)$.
The cognitive user partially uses its power to help send the
codewords of the primary user. $X_2$ contains two independent Gaussian parts,
$X_2 = \sqrt{\alpha P_2}V(m_1)+\sqrt{\bar{\alpha} P_2}U(m_2)$.
The primary user dedicates its whole power to transmit $m_1$, as
$X_1 = \sqrt{P_1}V(m_1)$.
Decoding in the primary receiver is based on successive
cancelation, where it first decode and subtract the cognitive
user's codeword in order to decode its own codeword free of interference.
Cognitive receiver simply decodes its
own codeword assuming the other codeword as interference.
\end{proof}

This achievable region even simplifies when the interference gain is
substantially large, that is $|a|\geq \sqrt{ P_1P_2} + \sqrt{1 + P_1 + P_1P_2}$. 
In the following section we show that,
when interference is very strong ($|a|\geq \sqrt{1 + P_1}$), the inner bound 
in Lemma~\ref{lem2} gives the capacity of the desired channel.

\subsection{Capacity of the GCZIC at very strong interference}
\label{capacity}

In this section, we prove that superposition coding
achieves the capacity of the GCZIC for $|a| \geq \sqrt{1 + P_1}$.
Specifically, we show that in Corollary~\ref{cor3}, $\beta$ has to be $1$
for this range of $a$.
Hence, the outer bound coincides with
the inner bound in Lemma~\ref{lem2} and gives
the capacity of the desired channel.
\begin{thm}
The capacity region the GCZIC for $|a|\geq \sqrt{1 + P_1}$
is the set of all rate pairs $(R_1, R_{2})$ satisfying
 \begin{align}
    R_2 &\leq   \frac{1}{2}\log{\Big(1 + \frac{\bar{\alpha}P_2}{1+\alpha P_2}\Big)}
    \label{eq:cap1} \\
  R_1 + R_2 & \leq \frac{1}{2}\log{\Big(1+(\sqrt{P_1} + |a| \sqrt{\alpha P_2}}\Big) +
  \frac{1}{2}\log{\Big(1 + \frac{\bar{\alpha}P_2}{1+\alpha P_2}\Big)}
   \label{eq:cap2}\\
   R_1 + R_2 &\leq \frac{1}{2}\log{\Big(1+ P_1 + a^2 P_2 + 2 |a| \sqrt{\alpha P_1P
      _2}\Big)} \label{eq:cap3}
  \end{align}
for $\alpha \in [0, 1]$.
\label{thm3}
\end{thm}

\begin{proof}
We want to show that for $|a|\geq \sqrt{1 + P_1}$ the outer bound in
Corollary~\ref{cor3} and the inner bound in Lemma~\ref{lem2} coincide.
Let's define $\mathcal R_{i}$ as the set of all rate pairs satisfying the
constraints in Lemma~\ref{lem2} and $\mathcal R_{o}$ as the set
of all rate pairs satisfying \eqref{eq:cap1}-\eqref{eq:cap3}.
Using the same argument as El Gamal \cite{ElGamalMoreCapable}, we
 can show that $\mathcal R_{i} \equiv \mathcal R_{o}$.
Here $R_{2}$ can be thought of as the rate of common message which can be decoded at both receivers
while $R_{1}$ is the rate of the private message. Now $(R_{2}, R_{1}) \in \mathcal R_{o} $
if and only if $(R_{2}-t, R_{1}+t) \in \mathcal R_{o} $ for any
$0 \leq t \leq R_{2}$. This means the common rate ($R_{2}$) can be
partially or entirely private. Thus region $\mathcal R_{o}$
can be represented as $\mathcal R_{i}$.

We next show that the outer bound in Corollary~\ref{cor3} simplifies
to $\mathcal R_{o}$. To do so, it suffices to show
that for $|a|\geq \sqrt{1 + P_1}$, $\beta$ has to be $1$ in Corollary~\ref{cor3}.
Consider the first two inequalities of the outer bound in Corollary~\ref{cor3};
we can see that on the boundary of this outer bound we must have
\begin{align}
   R_2 \leq  \frac{1}{2}\log{\Big(1+(\sqrt{\beta P_1} + |a| \sqrt{\alpha P_2}
    )^2\Big)}
     \label{eq:O07}
  \end{align}
Comparing this inequality with the first inequity in $\mathcal R_{i}$, we conclude
 that either $\beta$ has to be $1$ or the first inequity of $\mathcal R_{i}$ must be
loose; since otherwise the outer bound is less that the inner bound, which is not possible.
 For this inequality to be redundant in $\mathcal R_{i}$ we need
 \begin{align}
     \frac{1}{2}\log{\Big(1+ P_1 + a^2 P_2 + 2 |a| \sqrt{\alpha P_1P_2}\Big)} &<
     \frac{1}{2}\log{\left (1+ \frac{\bar{\alpha} P_2}{1 + \alpha P_2} \right)}+
      \frac{1}{2}\log{ \left (1 + (\sqrt{P_1} + |a| \sqrt{\alpha P_2 })^2 \right)}
      \label{eq:E8} \nonumber \\
       \Longleftrightarrow
     \frac{1}{2}\log{\Big(1+  \frac{a^2 \bar{\alpha} P_2 }{1 + (\sqrt{P_1} + |a| \sqrt{\alpha P_2 })^2} \Big)}
      &< \frac{1}{2}\log{\left (1+ \frac{\bar{\alpha} P_2}{1 + \alpha P_2} \right)}  \nonumber \\
        \Longleftrightarrow |a| &< \sqrt{\alpha P_1P_2} + \sqrt{1 + P_1 +\alpha P_1P_2}
 \end{align}
For \eqref{eq:E8} to hold with any $\alpha$ then
\begin{align}
     |a|< \sqrt{ 1+ P_1 }.
      \label{eq:E9}
 \end{align}
This implies the first inequality of Lemma~\ref{lem2} is redundant only
for $|a|< \sqrt{ 1+ P_1 }$. In other words, if $|a|\geq \sqrt{ 1+ P_1 }$ there
exist some $\alpha$ for which this inequality cannot be redundant; this in turn
enforces $\beta = 1$. Thus, for this range of $a$, the outer bound in 
Corollary~\ref{cor3} simplifies to $\mathcal R_{o}$.
Note also that with $\beta = 1 \, (\rho_1=0)$ 
the optimal input for the outer bound is the same as the input for the
inner bound (i.e., $X_1$, $U$ are independent and $X_2 =x_2(U,X_1$),
 thus the capacity is established.

%Next, we need to show that the convex hull of the rate regions defined by
%$\mathcal R_{i}$ and $\mathcal R_{o}$ is the same.
%To this end, we consider following two disjoint cases for a given
%$\alpha$. First, consider the case where \eqref{eq:cap2} is slack
%in $\mathcal R_{o}$, that is when \eqref{eq:E8} holds. Obviously, this implies
%that the first inequality in Lemma~\ref{lem2} is also slack; thus both
%$\mathcal R_{in}$ and $\mathcal R_{o}$ are defined with the same constraints
%and their equivalency is established. Alternatively, consider the case
%that, \eqref{eq:cap2} is binding in $\mathcal R_{o}$, or equivalently
%$|a| \geq \sqrt{\alpha P_1P_2} + \sqrt{1 + P_1 +\alpha P_1P_2}$. Clearly, this
%means that \eqref{eq:cap3}, the other sum rate, is slack in $\mathcal R_{o}$;
%then, it simply follows that the same constraint in $\mathcal R_{i}$,
%the last one, must be slack. Thus, the last constraint does not have any role in
%either of the regions.
%Now, for the same value of $\alpha$,
%any point on the boundary of the region defined by the first two constraint
%in $\mathcal R_{i}$ is on the boundary of the region defined by the first
%two constraint in $\mathcal R_{o}$, and vice versa. Therefore, the convex
%hull of the rate region defined by $\mathcal R_{i}$ is always equivalent
%to that of $\mathcal R_{o}$, and capacity region in Theorem \ref{thm3}
% is established.

\end{proof}

As a special case, when the third constraint is redundant
in $\mathcal R_{i}$ and $\mathcal R_{o}$, the outer bound in Corollary~\ref{cor2} is tight.
For such an $a$, the capacity region in Theorem~\ref{thm3} further simplifies as below.

\begin{cor}
The capacity region the GCZIC for $|a|\geq \sqrt{ P_1P_2} + \sqrt{1 + P_1 + P_1P_2}$
is the set of all rate pairs $(R_1, R_{2})$ satisfying
 \begin{align}
    R_1 &\leq  \frac{1}{2}\log{\Big(1+(\sqrt{P_1} + |a| \sqrt{\alpha P_2}
    )^2\Big)}
  %\label{eq:cognitive_rate_3} \\
  \nonumber \\
  R_2 &\leq \frac{1}{2}\log{\Big(1 +
    \frac{\bar{\alpha}P_2}{1+\alpha P_2}\Big)}
  \label{eq:capacit1}
  \end{align}
for $\alpha \in [0, 1]$.
\label{cor4}
\end{cor}

\begin{proof}
First consider the achievable rate region in Lemma~\ref{lem2}.
The third inequality becomes redundant if
 \begin{align}
     \frac{1}{2}\log{\Big(1+ P_1 + a^2 P_2 + 2 |a| \sqrt{\alpha P_1P_2}\Big)} &\geq
     \frac{1}{2}\log{\left (1+ \frac{\bar{\alpha} P_2}{1 + \alpha P_2} \right)}+
      \frac{1}{2}\log{ \left (1 + (\sqrt{P_1} + |a| \sqrt{\alpha P_2 })^2 \right)}
      \label{eq:E5} \nonumber \\
       \Longleftrightarrow
     \frac{1}{2}\log{\Big(1+  \frac{a^2 \bar{\alpha} P_2 }{1 + (\sqrt{P_1} + |a| \sqrt{\alpha P_2 })^2} \Big)}
      &\geq \frac{1}{2}\log{\left (1+ \frac{\bar{\alpha} P_2}{1 + \alpha P_2} \right)}  \nonumber \\
        \Longleftrightarrow |a| &\geq \sqrt{\alpha P_1P_2} + \sqrt{1 + P_1 +\alpha P_1P_2}
 \end{align}
For \eqref{eq:E5} to hold with any $\alpha$ then
\begin{align}
     |a|\geq \sqrt{ P_1P_2} + \sqrt{1 + P_1 + P_1P_2}.
      \label{eq:E6}
 \end{align}
On the other hand, without the third inequality,
the achievable region in Corollary~\ref{cor4} is also equal to the
outer bound in the Corollary~\ref{cor2}. This is because, for
the same value of $\alpha$, any point on the boundary of this region
 is on the boundary of the region in Corollary~\ref{cor2}, and vice versa.
 Hence, we obtain the capacity region in Corollary~\ref{cor4}
 if (\ref{eq:E6}) holds for any  $\alpha$, i.e.,
 $|a| \geq \sqrt{P_1P_2} + \sqrt{1 + P_1 + P_1P_2}$.
\end{proof}

Corollary~\ref{cor4} provides a special case of Theorem~\ref{thm3} with
simpler rate expressions.
In the concurrent and independent work \cite{Rini2},
``Theorem V.3. Capacity for S-G-CIFC" archives the same capacity result
in Corollary~\ref{cor4} using a different approach. The achievability
follows from a more general DPC-based scheme for the Gaussian cognitive
interference channel. Although the achievability seems
to be based on DPC, a close observation reveals that DPC is not necessary
to achieve the capacity.
This is because the parameter of DPC is zero ($\lambda = 0$) which means that, in effect,
DPC has no contribution; thus, the achievability scheme reduces
to superposition coding. The outer bound is completely different and is
based on the MIMO-BC outer bound \cite{Rini2}.

Theorem~\ref{thm3} shows that, when the interference is very strong,
the interfered primary receiver can decode the message of the interfering
cognitive transmitter in a rate higher than its own receiver.
This is as though the interference
link from the cognitive transmitter to the primary receiver is less noisy
than the direct link for the cognitive user's message. This sheds light on
the optimal coding scheme as well. That is, the primary user
encodes independently while the cognitive user encodes by superimposing
the primary user's codeword on its own. Then, the cognitive receiver decodes
its message treating primary user's codeword as noise. The primary receiver,
on the other hand, performs successive cancellation; it first decodes
the cognitive user's message, then subtracts it from the
received signal to decode its own message free of interference.

\section{Summary}
\label{sec:sum}

Analysis of the capacity results of the GCZIC shows that
superposition coding appears to be an indispensable tool
in any capacity-achieving techniques for this channel.
At very strong interference, superposition coding single-handedly
achieves its capacity. However, both DPC and superposition coding
are needed to establish the capacity when interference is
weak or intermediate.
Table I summarizes the capacity results for the GCZIC.
As it can be seen, up to now, the capacity of this channel is characterized except
for $\sqrt { 1 + \frac{P1}{1+P2}} < a < \sqrt{1 + P_1 } $.
For this range of $a$, a more general and inclusive form of
the proposed capacity-achieving outer bounds at $|a|\geq 1$, as represented
in Lemma \ref{lem1}, provides the best outer bound on the capacity region of the GCZIC.

\begin{table}[!t]
\renewcommand{\arraystretch}{1.3}
\caption{Summary of all capacity results for the GCZIC} \label{table1}
\centering
\begin{tabular}{|c|c|c|c|}
\hline
\bfseries Range of \textit{a} & \bfseries Capacity region & \bfseries Capacity achieving  & \bfseries Reference \\
\bfseries   & & \bfseries technique &  \\
\hline\hline\
 $|a| \leq 1$ & $
   R_{1} \leq \frac{1}{2}\log{\Big(1 + \frac{(\sqrt{P_1} + |a|
       \sqrt{\alpha  P_2 })^2} {1+ a^2 \bar{\alpha} P_2}}\Big)
 $ & superposition coding  & \cite{Jovicic-Viswanath}, \cite{Wu-Vishwanath} \\
 $  $ & $ R_2 \leq \frac{1}{2}\log{(1+\bar{\alpha} P_2)} $ & and DPC   &  \\
\hline

 $1 \leq |a| \leq  \sqrt {1 + \frac{P1}{1+P2}}$ & $
   R_1 + R_2 \leq \frac{1}{2}\log{ \Big(1+P_1 + a^2 P_2 + 2|a| \sqrt{\alpha P_1P_2}
   \Big) } $ & superposition coding &  \cite{Rini}, \cite{JiangZ} \\
 $  $ & $ R_2 \leq \frac{1}{2}\log{(1+\bar{\alpha} P_2)} $ &  and DPC &  \\
\hline

$ |a| > \sqrt { 1 + \frac{P1}{1+P2}} $ & unknown & unknown  & ---\\
$ |a| < \sqrt{1 + P_1 } $ & (Lemma~\ref{lem1} gives the best outer bound) &   &\\
\hline

$  $ & $  R_1\leq \frac{1}{2}\log{ \left (1 + (\sqrt{P_1} + |a| \sqrt{\alpha P_2 })^2 \right)}
$  & & \\
 $ |a| \geq \sqrt {1 + P1} $ & $  R_2 \leq \frac{1}{2}\log{\Big(1 +
    \frac{\bar{\alpha}P_2}{1+\alpha P_2}\Big)}$ & superposition coding  & Theorem [\ref{thm3}]  \\
$  $ & $ R_1 + R_2 \leq \frac{1}{2}\log{ \Big(1+P_1 + a^2 P_2 + 2|a| \sqrt{\alpha P_1P_2}
   \Big) }$ &  &  \\
\hline

$ |a| \geq \sqrt{P_1P_2} + \sqrt{1 + P_1 + P_1P_2}$ & $  R_1\leq \frac{1}{2}\log{ \left (1 + (\sqrt{P_1} + |a| \sqrt{\alpha P_2 })^2 \right)}
$ & superposition coding  & Corollary~[\ref{cor4}],  \\
 $  $ & $  R_2 \leq \frac{1}{2}\log{\left (1+ \frac{\bar{\alpha} P_2}{1 + \alpha P_2} \right)} $ &  & \cite{Rini2}  \\
\hline
\end{tabular}
\end{table}

\newpage

\section{Appendix}
\label{anx}
\subsection{Proof of Theorem \ref{thm1}}

\begin{proof}
We prove this theorem by showing the encoding, decoding,
and error analysis.

\subsubsection{Code construction and encoding}
Fix $p(x_1)$, $p(u)$ and $p(x_2|x_1, u)$ that achieve capacity. Randomly and
independently generate $2^{nR_{1}}$ sequences $x_1^n(m_{1})$, $m_{1}
\in [1: 2^{nR_{1}}]$ {\textit{i.i.d.}} according to $\prod_{i=1}^np(x_{1i})$. Also,
randomly and independently generate $2^{nR_{2}}$ sequences
$u^n(m_{2})$, $m_{2} \in [1: 2^{nR_{2}}]$ with elements {\textit{i.i.d.}}
according to $\prod_{i=1}^np_U(u_i)$. Next, for each pair of sequences $(u^n(m_2), x_1^n(m_1))$,
randomly and conditionally independently generate one
sequence $x_2^n(m_1, m_2)$ with elements {\textit{i.i.d.}} according to
$\prod_{i=1}^np_{X_2|X_1U}(x_{2i}|x_{1i}(m_1)u_{i}(m_2))$.

For encoding, to send messages $(m_1,m_2)$, the primary transmitter
just sends the codeword $x_1^n(m_1)$  and the secondary transmitter
sends the codeword $x_2^n(m_1,m_2)$ corresponding to those messages.

\subsubsection{Decoding}
Decoding is based on standard joint typicality. The less capable
receiver ($Y_2$) can only distinguish the auxiliary random variable
$U$. Decoder 2 declares that message $\hat{m}_{2}$ is sent if it is
a unique message such that $(u^n(\hat{m}_2), y^n_2) \in {\cal
T}^{(n)}_\epsilon$ ; otherwise it declares an error. Decoder 1
declares that message $\hat{m}_{1}$ is sent if it is a unique
message such that $(u^n(m_2), x_1^n(\hat{m}_1),
x_2^n(\hat{m}_1,m_2), y^n_2) \in {\cal T}^{(n)}_\epsilon$ for some
$m_2$; otherwise it declares an error.

\subsubsection{Error analysis}
To analyze the probability of error, without loss of generality, assume that
$(M_1,M_2)=(1,1)$ is sent. First we consider the average probability of
error for decoder 2. Let's define the error events
\begin{align}
 {\cal E}_{21} &= (U^n(1), Y^n_2) \notin {\cal T}^{(n)}_\epsilon \nonumber \\
 {\cal E}_{22} &= (U^n(m_2), Y^n_2) \in {\cal T}^{(n)}_\epsilon  \text { for some } m_2 \neq 1
  \label{eq:E2}
\end{align}
By union bound, the probability of error for decoder 2 is upper bounded by
\begin{align}
 {\cal E}_{2} = P({\cal E}_{21} \cup {\cal E}_{22} ) \leq P({\cal E}_{21}) + P({\cal E}_{22} ).
\end{align}
Now by law of large numbers (LLN) $ P({\cal E}_{21}) \rightarrow
0 \text { as }  n \rightarrow \infty. $ Also, since $U^n(m_2)$
is independent of $(U^n(1), Y^n_2)$ for $ m_2 \neq 1 $ by the packing
lemma \cite{ElGamal}
 $ P({\cal E}_{22}) \rightarrow 0 \text { as }  n \rightarrow \infty \text { if }  R_2 \leq I(U;Y_2)-\delta(\epsilon). $

Then, consider the average probability of error for decoder 1.
We define the following error events
\begin{align}
 {\cal E}_{11} =& (X_1^n(1), U^n(1), X_2^n(1,1), Y^n_1) \notin {\cal T}^{(n)}_\epsilon \nonumber \\
 {\cal E}_{12} =& (X_1^n(m_1), U^n(1), X_2^n(m_1,1), Y^n_1) \in {\cal
 T}^{(n)}_\epsilon \nonumber \\
 &\text { for some }  m_1 \neq 1  \nonumber \\
 {\cal E}_{13} =& (X_1^n(m_1), U^n(m_2), X_2^n(m_1,m_2), Y^n_1) \in {\cal T}^{(n)}_\epsilon  \nonumber \\
 &\text { for some } m_1 \neq 1, m_2 \neq 1.
  \label{eq:E1}
\end{align}
Using union bound, the probability of error for decoder 2 is upper bounded by
\begin{align}
 {\cal E}_{2} = P({\cal E}_{11} \cup {\cal E}_{12} \cup {\cal E}_{13} )
 \leq P({\cal E}_{11}) + P({\cal E}_{12}) +P({\cal E}_{13} )
\end{align}
Now we evaluate each term in the right-hand side (RHS) of this inequality when $ n \rightarrow
\infty$. First consider ${\cal E}_{11}$; again by LLN
$ P({\cal E}_{11}) \rightarrow 0 \text { as }  n \rightarrow
\infty. $ Next consider ${\cal E}_{12}$. For $ m_1 \neq 1 $ since
$(X_1(m_1), X_2(m_1,1))$ is conditionally independent of $Y^n_1$
given $U^n(1)$, by packing lemma
 $ P({\cal E}_{12}) \rightarrow 0 \text { as }  n \rightarrow
 \infty \text { given }  R_1 \leq I(X_1, X_2;Y_1|U)-\delta(\epsilon) = I(X_1;Y_1|U)-\delta(\epsilon)$
 because $X_2$ is a function of $(U, X_1)$.
Finally consider ${\cal E}_{13}$. For $ m_1 \neq 1 $ and $ m_2 \neq 1 $,
 $(X_1(m_1), U^n(m_2), X_2(m_1,m_2))$ is independent of
$Y^n_1$; hence, by packing lemma
 $ P({\cal E}_{13}) \rightarrow 0 \text { as }
 n \rightarrow \infty \text { if  }  R_1 + R_2 \leq
 I(X_1,U,X_2;Y_1)-\delta(\epsilon) = I(X_1,X_2;Y_1)
  -\delta(\epsilon)$. The equality follows because $U \rightarrow X_1,X_2 \rightarrow Y_1$
  forms a Markov chain.
The above analysis completes the proof of achievability since it shows that both receivers
can decode corresponding messages with the total probability of error tending to
zero if  \eqref{eq:inner} is satisfied. Therefore,
there exists a sequence of codes with error probability tending
to 0.
\end{proof}

\subsection{Proof of Theorem \ref{thm2}}
%\begin{proof}
\subsubsection{For more capable DM-CIC}

The proof is also similar to the converse proof for the more
capable broadcast channel \cite{ElGamalMoreCapable}. We follow the
same line of proof as in \cite{ElGamal}; the only difference is
replacing $X_1$ in \cite{ElGamal} with $(X_1, X_2)$, since here
$X_2$ also encodes $M_1$.

We can bound the rates $R_2$ and $R_2 + R_2$ as
\begin{align}
  nR_2 & = H(M_2)  \nonumber \\
  & = I(M_2;Y^n_2) + H(M_2|Y^n_2)  \nonumber \\
  & \leq I(M_2;Y^n_2) + n\epsilon_{1n}  \label{eq:F-1}
  \end{align}
and
\begin{align}
  n(R_1 +  R_2) & =  H(M_1, M_2) \nonumber \\
  & =  H(M_1| M_2)+ H(M_2) \nonumber \\
  & =  I(M_1;Y^n_1|M_2) + H(M_1|Y^n_1,M_2) +  I(M_2;Y^n_2)  + H(M_2|Y^n_2) \nonumber \\
  & \leq I(M_1;Y^n_1|M_2) + I(M_2;Y^n_2)  + n\epsilon_{2n} \label{eq:F-2}
 \end{align}
 where \eqref{eq:F-1} and (\ref{eq:F-2}) follow by Fano's inequality. In a very similar fashion, sum rate can be also bounded by
  \begin{align}
  n(R_1 +  R_2) & \leq I(M_2;Y^n_2|M_1) + I(M_1;Y^n_1) + n\epsilon_{3n}. \label{eq:F-3}
 \end{align}
Now we manipulate the RHS of \eqref{eq:F-1}-\eqref{eq:F-3} to obtain the desired terms in \eqref{eq:Outer}.
First, consider the mutual information term in \eqref{eq:F-1}
\begin{align}
  I(M_2;Y^n_2)  & =\sum^n_{\substack{i=1}}I(M_2;Y_{2i}|Y^n_{2,i+1}) \label{eq:R2-1} \\
  & \leq \sum^n_{\substack{i=1}}I(M_2,Y^n_{2,i+1};Y_{2i})\nonumber \\
  & \leq \sum^n_{\substack{i=1}}I(M_2, Y^{i-1}_{1}, Y^n_{2,i+1};Y_{2i}) \nonumber \label{eq:R2-2} \\
  & =\sum^n_{\substack{i=1}}I(U_i;Y_{2i})
  \end{align}
in which (\ref{eq:R2-1}) follows from the chain rule,
 and we have defined the auxiliary random variable  $U_i
= (M_2, Y^{i-1}_{1}, Y^n_{2,i+1})$ moving to (\ref{eq:R2-2}).
Next, we bound the mutual information terms of the second inequality in (\ref{eq:F-2}).
\begin{align}
  I(M_1;Y^n_1|M_2) + I(M_2;Y^n_2)
  =& \sum^n_{\substack{i=1}}I(M_1; Y_{1i}|M_2,Y^{i-1}_{1}) + \sum^n_{\substack{i=1}}I(M_2;Y_{2i}|Y^n_{2,i+1})\nonumber \\
  \leq& \sum^n_{\substack{i=1}}I(M_1, Y^n_{2,i+1}; Y_{1i}|M_2,Y^{i-1}_{1}) + \sum^n_{\substack{i=1}}I(M_2, Y^n_{2,i+1};Y_{2i})\nonumber \\
  =& \sum^n_{\substack{i=1}}I(M_1, Y^n_{2,i+1}; Y_{1i}|M_2,Y^{i-1}_{1})   \nonumber \\ &
  + \sum^n_{\substack{i=1}}I(M_2, Y^n_{2,i+1},Y^{i-1}_{1};Y_{2i}) - \sum^n_{\substack{i=1}}I(Y^{i-1}_{1};Y_{2i}|M_2,Y^n_{2,i+1})\nonumber \\
  =& \sum^n_{\substack{i=1}}I(M_1;  Y_{1i}|M_2,Y^{i-1}_{1},Y^n_{2,i+1}) +
   \sum^n_{\substack{i=1}}I(M_2, Y^n_{2,i+1},Y^{i-1}_{1};Y_{2i}) \nonumber \\ & - \sum^n_{\substack{i=1}}I(Y^{i-1}_{1};Y_{2i},|M_2,Y^n_{2,i+1}) %\nonumber \\ &
   + \sum^n_{\substack{i=1}}I(Y^n_{2,i+1};Y_{1i},|M_2,Y^{i-1}_{1})\nonumber
  \\  \label{eq:sum}
  =& \sum^n_{\substack{i=1}}I(M_1; Y_{1i}|U_{i}) +  \sum^n_{\substack{i=1}}I(U_i;Y_{2i}) \\
  \leq&  \sum^n_{\substack{i=1}}I(X_{1i}; Y_{1i}|U_{i}) + \sum^n_{\substack{i=1}}I(U_i;Y_{2i})
  \label{eq:Rsum1}
\end{align}
 where (\ref{eq:sum}) follows by the Csiszar sum identity and the
auxiliary random variable $U_i = (M_2, Y^{i-1}_{1}, Y^n_{2,i+1})$;
(\ref{eq:Rsum1}) follows by Markov chain $M_1 \rightarrow (X_1 U)
\rightarrow Y_1$.

For the third inequality, following steps similar to the bound for the second
inequality and defining $V_i := (M_1, Y^{i-1}_1,Y^n_{2,i+1})$ we can bound the
 mutual information terms  in (\ref{eq:F-3}) as

\begin{align}
   I(M_2;Y^n_2|M_1) + I(M_1;Y^n_1)
  & = \sum^n_{\substack{i=1}}I(M_{2}; Y_{2i}|M_1,Y_{2}^{i-1}) + \sum^n_{\substack{i=1}}I(M_1;Y_{2i}|Y_{2}^{i-1})  \nonumber \\
  \label{eq:a0}
  & \leq \sum^n_{\substack{i=1}}I(M_{2}; Y_{2i}|V_{i}) + \sum^n_{\substack{i=1}}I(V_{i};Y_{2i})  \\
  \label{eq:a1}
   & \leq \sum^n_{\substack{i=1}}I(X_{1i},X_{2i}; Y_{2i}|V_{i}) + \sum^n_{\substack{i=1}}I(V_i;Y_{1i})  \\
  \label{eq:b}
    & \leq \sum^n_{\substack{i=1}}I(X_{1i},X_{2i}; Y_{1i}|V_{i}) + \sum^n_{\substack{i=1}}I(V_i;Y_{1i}) \\
       % \label{eq:b1}
  & = \sum^n_{\substack{i=1}}I(X_{1i},X_{2i}; Y_{1i}) \nonumber
\end{align}
in which (\ref{eq:a0}) follows similar steps to the bound for the
first inequality on sum rate; (\ref{eq:a1}) follows from  $M_2 \rightarrow X_1, X_2
\rightarrow Y_2$; and (\ref{eq:b}) follows by (\ref{eq:defn1}) that gives
$I(X_1,X_2;Y_1) \geq I(X_1,X_2;Y_2)$, and implies that
$I(X_1,X_2; Y_2|V) \leq I(X_1,X_2; Y_1|V)$.

Next, we define the time sharing random variable $Q$
which is uniformly distributed over $[1:n]$ and is independent of
$(M_1,M_2,X^n_1,X^n_2,Y^n_1,Y^n_2)$. Also we define
$U=(Q,U_Q),X^n_1 = X^n_{1Q},X^n_2 =X^n_{2Q},Y^n_1 =Y^n_{1Q},Y^n_2=Y^n_{2Q}$. Then we have

\begin{align*}
  nR_2 &\leq \sum^n_{\substack{i=1}}I(U_i;Y_{2i}) + n\epsilon_{1n} \\
  & =nI(U;Y_{2}|Q) + n\epsilon_{1n}\\
  & \leq nI(U;Y_{2}) + n\epsilon_{1n}
\end{align*}
Similarly
\begin{align*}
  n(R_1 +  R_2) &\leq \sum^n_{\substack{i=1}}I(X_{1i}; Y_{1i}|U_{i}) + \sum^n_{\substack{i=1}}I(U_i;Y_{2i}) + n\epsilon_{2n} \\
  & =nI(X_1;Y_1|U,Q) + nI(U;Y_{2}|Q) + n\epsilon_{2n}\\
  & =nI(X_1;Y_1|U) + nI(U;Y_{2}|Q) + n\epsilon_{2n}\\
  &\leq nI(X_1;Y_1|U) + nI(U;Y_{2}) + n\epsilon_{2n}
 \end{align*}
 and
 \begin{align*}
  n(R_1 +  R_2) &\leq \sum^n_{\substack{i=1}}I(X_{1i},X_{2i}; Y_{1i})  + n\epsilon_{3n}.\\
   & =nI(X_1,X_1;Y_1|Q) + n\epsilon_{3n}\\
   &\leq nI(X_1,X_1;Y_1) + n\epsilon_{3n}
 \end{align*}

As $n \rightarrow \infty $, $\epsilon_{1n},\epsilon_{2n},$ and $ \epsilon_{3n}$
tend to zero because the probability of error is assumed to vanish.
This completes the proof for the more capable DM-CIC.

\subsubsection{For DM-CIC with strong interference}
\begin{proof}
The proof outer bound under the strong interference condition \eqref{eq:defn2} is
almost the same, with only slight difference in the proof of last inequality.
This is because the first two inequalities hold for any DM-CIC. Under
the strong interference condition \eqref{eq:defn2} we can we can bound the
 mutual information terms  in (\ref{eq:F-3}) as
\begin{align}
   I(M_2;Y^n_2|M_1) + I(M_1;Y^n_1)
  & \leq   I(M_2;Y^n_2|M_1,X^n_1) + I(M_1,X^n_1;Y^n_1) \nonumber\\
  & \leq   I(X^n_2;Y^n_2|M_1,X^n_1) + I(M_1,X^n_1;Y^n_1)  \nonumber\\
  & \leq   I(X^n_2;Y^n_1|M_1,X^n_1) + I(M_1,X^n_1;Y^n_1)  \label{eq:a6}\\
  & =  I(X^n_1,X^n_2;Y^n_1) \nonumber\\
  & \leq \sum^n_{\substack{i=1}}I(X_{1i},X_{2i}; Y_{1i})\nonumber
\end{align}
in which (\ref{eq:a6}) follows by (\ref{eq:defn2}) that gives
$I(X_2;Y_2|X_1) \leq I(X_2;Y_1|X_1)$, and implies that
$I(X_2;Y_2|X_1,M_1) \leq I(X_2;Y_1|X_1,M_1)$. The other steps are
straightforward.
Finally, this proves that Theorem~\ref{thm2} holds both for more capable
strong and interference DM-CIC.
\end{proof}

\subsection{Proof of Lemma \ref{lem1}}
\begin{proof}
We need to find the distribution that maximize the rate region in Theorem \ref{thm2} for
 the Gaussian channel. In what follows, we show that jointly Gaussian $X_1, X_2, U $ is optimum, i.e.,
it provides the largest outer bound for the Gaussian channel.
By maximum entropy theorem, the RHS of the third inequality is maximized
when $Y_1$ is Gaussian, thus
$Y_1= X_1 + a X_2 + Z_1  \text{ is Gaussian; that is } X_1 + a X_2 \text{ must be Gaussian. }$

Similarly,
\begin{align}
R_{2}&\leq I(U;Y_2)  \nonumber\\
&= h(Y_2) - h(Y_2|U) \nonumber\\
&\leq h(Y_2^G) - h(Y_2|U) \nonumber\\
&\leq h(Y_2^G)- \frac{1}{2}\log \left({2^{2h(X_2|U)} + 2^{2h(Z_2)}} \right )
\end{align}
where $Y_2^G$ denotes $Y_2$ when the inputs $X_1, X_2$ are Gaussian.
The last inequality follows by conditional version of entropy power inequality (EPI) for which
equality is achieved when $X_2|U  \text{ is Gaussian }$.
Likewise, for the term $I(X_1;Y_1|U)$ we can write
\begin{align*}
I(X_1;Y_1|U) &= h(Y_1|U) - h(Y_1|U,X_1)\\
&\leq h(X_1^G|U^G + a X_2^G|U^G + Z_1)- h(Y_1|U,X_1)\\
&\leq h(X_1^G|U^G + a X_2^G|U^G + Z_1)- \frac{1}{2}\log \left({2^{2h(aX_2|U,X_1)} + 2^{2h(Z_2)}} \right )
\end{align*}
where where $X_1^G, U^G$ denote$X_1, U$ when the inputs $X_1, X_2$ are Gaussian,
and inequalities follow by maximum entropy theorem and conditional version of EPI, respectively.
Again equality is achieved when all terms are Gaussian. Hence,
all inequalities in the outer bound are maximized with
jointly Gaussian $X_1, X_2, U $.

Now the problem is to find the optimum covariance matrix to maximize the bounds,
 i.e., to determine correlation coefficients
among $X_1, X_2,$ and $ U$.
Let $(U, X_1, X_2) \sim {\cal N} (0, K)$, which are correlated Gaussian random variables
$X_1, X_2,$ and $ U$ with covariance matrix
\begin{align}
K=COV(U, X_1, X_2)
 = \left( \begin{array}{ccccccc}
        P_U &  \rho_{1} \sqrt{P_U P_1} & \rho_{2} \sqrt{P_U P_2} \\
        \rho_{1} \sqrt{P_U P_1} &  P_1 & \rho_{12} \sqrt{P_1 P_2} \\
        \rho_{2} \sqrt{P_U P_2}  & \rho_{12} \sqrt{P_1 P_2} &  P_2 \\
      \end{array}\right)  \label{eq:cov}
\end{align}

\noindent Since the covariance matrix is positive semidefinite, the determinant of this
matrix must be nonnegative. That is
 \begin{align}
1-\rho_{12}^2-\rho_{1}^2+2\rho_{1}\rho_{2}\rho_{12}- \rho_{2}^2\geq 0.
 \end{align}
The inequality holds if
\begin{align}
\rho_{1}\rho_{2}- \sqrt {(1-\rho_{1}^2 )(1-\rho_{2}^2 )}  \leq  \rho_{12}
 \leq  \rho_{1}\rho_{2}  + \sqrt {(1-\rho_{1}^2 )(1-\rho_{2}^2 )}\nonumber
   \end{align}
or equivalently the covariance matrix is positive semidefinite if
\begin{align}
  |\rho_{12} -  \rho_{1}\rho_{2}|
 \leq  \sqrt {(1-\rho_{1}^2 )(1-\rho_{2}^2 )}. \label{eq:cond4}
\end{align}
Now we evaluate the rate constraints defining the outer bound in Theorem \ref{thm2}.
\begin{align}
R_{2}&\leq I(U;Y_2) \nonumber \\
&= h(Y_2) +  h(U) - h(Y_2,U)\nonumber \\
&= h(X_2 + Z_2) +  h(U) - h(X_2 + Z_2,U )\nonumber \\
&= \frac{1}{2}\log{2 \pi e \left(1 + P_2 \right)} +  \frac{1}{2}\log{2 \pi e \left( P_U \right)}
- \frac{1}{2}\log{(2 \pi e)^2 |K_1| }\nonumber \\
&= \frac{1}{2}\log{\left( \frac{1+ P_2}{ 1 +P_2(1- \rho_{2}^2)} \right)}
\end{align}
where
\begin{align}
K_1 = \left( \begin{array}{cc}
      1+ P_2 &   \rho_{2} \sqrt{P_U P_2}  \nonumber \\
        \rho_{2} \sqrt{P_U P_2}  & P_U   \nonumber \\
\end{array}\right) \nonumber
\end{align}
Similarly
\begin{align}
I(X_1;Y_1|U) &= h(Y_1|U) - h(Y_1|U,X_1) \nonumber \\
%&= h(X_1^G|U^G + a X_2^G|U^G + Z_1)- h(a X_2^G|U^G,X_1^G + Z_1)\\
&= h(Y_1,U) - h(U) - \frac{1}{2}\log \left({2^{2h(aX_2|U,X_1)} + 2^{2h(Z_1)}} \right ) \nonumber\\ \label{eq:ineq1}
&\leq h(Y_1,U) - h(U) - \frac{1}{2}\log \left({2^{2h(Z_1)}} \right )\\
&= \frac{1}{2}\log{\left( |K_2|/ P_U  \right)}\nonumber
\end{align}
where \begin{align}
|K_2| &= \left|\left( \begin{array}{cc}
       1+P_1 + a^2 P_2 + 2a \rho_{12} \sqrt{ P_1P_2} &  \rho_{1} \sqrt{ P_U P_1} + a \rho_{2} \sqrt{ P_U P_2} \\
       \rho_{1} \sqrt{ P_U P_1} + a \rho_{12} \sqrt{ P_U P_2}  & P_U   \\
\end{array}\right)\right|\\\nonumber
&= P_U (1+P_1 + a^2 P_2 + 2a \rho_{12} \sqrt{ P_1P_2}- \rho_{1}^2 P_1 - a^2\rho_{2}^2 P_2 -  2a \rho_{1} \rho_{2}  \sqrt{P_1P_2})
\end{align}
To check if the inequality (\ref{eq:ineq1}) can hold with equality, we evaluate the term $2^{2h(aX_2|U,X_1)}$
\begin{align}
h(X_2|U,X_1) &= h(U,X_1,X_2) - h(U,X_1)\nonumber \\
&= \frac{1}{2}\log \left((2\pi e)^3|K| \right) - \frac{1}{2}\log{\left( (2\pi e)^2 P_U P_1 ( 1- \rho_{1}^2 ) \right)}\nonumber \\
&= \frac{1}{2}\log{2\pi e P_2 \left( \frac{1- \rho_{12}^2 - \rho_{1}^2 -\rho_{2}^2 + 2 \rho_{12}\rho_{2} \rho_{1}}{1-  \rho_{1}^2} \right)}
\label{eq:ineq2}
\end{align}
in which the covariance matrix $K$ is defined by (\ref{eq:cov}).
Since both numerator and denominator are nonnegative in (\ref{eq:ineq1}), the argument of this function is either zero or positive. Therefore,
 $\min ({2^{2h(aX_2|U,X_1)}})=0$ is achieved when
$1- \rho_{12}^2 - \rho_{1}^2 -\rho_{2}^2 + 2 \rho_{12}\rho_{2} \rho_{1}=0$, or equivalently, $|\rho_{12} -  \rho_{1}\rho_{2}|
 = \sqrt {(1-\rho_{1}^2 )(1-\rho_{2}^2 )}$. Note that
$h(aX_2|U,X_1)=0$ implies  $X_2$ to be a function of $(U,X_1)$.

Keeping this in mind that $|\rho_{12} -  \rho_{1}\rho_{2}|
 =  \sqrt {(1-\rho_{1}^2 )(1-\rho_{2}^2 )}$ is optimum condition for (\ref{eq:ineq1}),
 we evaluate $I(X_1;Y_1|U)$ as follows.
\begin{align}
I(X_1;Y_1|U) & \leq  \frac{1}{2}\log{\left( |K_2|/ P_U  \right)} \nonumber \\
&= \frac{1}{2}\log{\left( 1+(1- \rho_{1}^2) P_1 + a^2 (1-\rho_{2}^2 ) P_2 +
2a (\rho_{12} -\rho_{1} \rho_{2} )\sqrt{ P_1P_2}  \right)} \nonumber \\
&\leq \frac{1}{2}\log{\left( 1+(1- \rho_{1}^2) P_1 + a^2 (1-\rho_{2}^2 ) P_2 +
2|a| |\rho_{12} -\rho_{1} \rho_{2}|\sqrt{ P_1P_2}  \right)}  \nonumber\\
&\leq \frac{1}{2}\log{\left( 1+(1- \rho_{1}^2) P_1 + a^2 (1-\rho_{2}^2 ) P_2 +
2|a|\sqrt {(1-\rho_{1}^2 )(1-\rho_{2}^2 )P_1P_2} \right)} \nonumber
\end{align}
where, the last inequality follows from (\ref{eq:cond4}).
Interestingly, again
\begin{align}
 |\rho_{12} -  \rho_{1}\rho_{2}|
=  \sqrt {(1-\rho_{1}^2 )(1-\rho_{2}^2 )} \label{eq:cond5}
\end{align}
turns out to be the optimum condition.
From this two values for $\rho_{12}$ are plausible which are respectively
\begin{align}
 \rho_{12}^{(1)} = \rho_{1}\rho_{2} +  \sqrt {(1-\rho_{1}^2 )(1-\rho_{2}^2 )} \\
 \rho_{12}^{(2)} = \rho_{1}\rho_{2} - \sqrt {(1-\rho_{1}^2 )(1-\rho_{2}^2 )}
\end{align}
Then, it is also straightforward to calculate the third bound in Theorem \ref{thm2}
 to obtain
 \begin{align}
    R_1 +R_{2} &\leq \frac{1}{2}\log{ \left( 1+P_1 + a^2 P_2 + 2a \rho_{12} \sqrt{P_1P_2} \right)}\nonumber \\
    &\leq \frac{1}{2}\log{ \left( 1+P_1 + a^2 P_2 + 2  \max\{a\rho_{12}^{(1)} , a\rho_{12}^{(2)}\} \sqrt{P_1P_2} \right)} \nonumber
    % &\leq \frac{1}{2}\log{ \left( 1+P_1 + a^2 P_2 + 2|a| (|\rho_{1}\rho_{2}| + \sqrt {(1-\rho_{1}^2 )(1-\rho_{2}^2 )}) \sqrt{P_1P_2} \right)}
 \end{align}

\noindent Therefore, from the outer bound in Theorem \ref{thm2}, we compute the rate region
%$\mathcal{R}_\mathcal{G}$ $\mathcal{R}_G$
$\mathcal{R}_\mathbb{G}$ with following constraints
\begin{align}
  R_2 &\leq  \frac{1}{2}\log{\left( \frac{1+ P_2}{ 1 +P_2(1- \rho_{2}^2)} \right)}
 \nonumber\\
  R_1 + R_2 &\leq  \frac{1}{2}\log{\left( 1+(1- \rho_{1}^2) P_1 + a^2 (1-\rho_{2}^2 ) P_2 +
2|a| \sqrt {(1-\rho_{1}^2 )(1-\rho_{2}^2 )P_1P_2} \right)} \nonumber \\
   &+  \frac{1}{2}\log{\left( \frac{1+ P_2}{ 1 +P_2(1- \rho_{2}^2)} \right)}
  \nonumber \\
    R_1 +R_{2} &\leq \frac{1}{2}\log{ \left( 1+P_1 + a^2 P_2 + 2  \max\{a\rho_{12}^{(1)} ,
     a\rho_{12}^{(2)}\} \sqrt{P_1P_2} \right)} \nonumber
        %  R_1 +R_{2} &\leq \frac{1}{2}\log{ \left( 1+P_1 + a^2 P_2 + 2|a| (|\rho_{1}\rho_{2}| + \sqrt {(1-\rho_{1}^2 )(1-\rho_{2}^2 )}) \sqrt{P_1P_2} \right)}
 \end{align}

\noindent As a last step, we can evaluate and add the Gaussian version of the standard
inequality $R_2 \leq I(X_2;Y_2|X_1)$ \cite {Wu-Vishwanath},
\cite {Maric1}, to these bounds; the corresponding inequality is
\begin{align}
     R_2 &\leq  \frac{1}{2}\log{\left(  1 +(1-\rho_{12}^2)P_2 \right)}
 \end{align}
As a result, the outer bound is as given in Lemma \ref{lem1}.
\end{proof}

\end{document}